\documentclass[11pt]{article}
\usepackage{amsmath,amsthm,amssymb,epsfig,sectsty,caption,color,times}
\usepackage[letterpaper,hmargin=0.95in,vmargin=1.1in]{geometry}

\usepackage[ruled]{algorithm2e}
\SetAlFnt{\small}
\SetAlCapFnt{\small}
\SetAlCapNameFnt{\small}
\SetAlCapHSkip{0pt}
\IncMargin{-\parindent}

\theoremstyle{plain}
\newtheorem{theorem}{Theorem}[section]
\newtheorem{lemma}[theorem]{Lemma}
\newtheorem{claim}[theorem]{Claim}

\theoremstyle{definition}
\newtheorem{definition}[theorem]{Definition}

\newcommand{\qedsymb}{\hfill{\rule{2mm}{2mm}}}
\renewenvironment{proof}{\begin{trivlist} \item[\hspace{\labelsep}{\bf \noindent Proof.\/}] }{\qedsymb\end{trivlist}}%
\newcommand{\bbN}{\mathbb{N}}

\newcommand{\name}[1]{#1}
\newcommand{\opt}{\mathrm{OPT}}
\newcommand{\alg}{\mathrm{ON}}

\sectionfont{\large} \subsectionfont{\normalsize}
\usepackage{ifthen}
\def\Context{PAPER}

\begin{document}

\begin{titlepage}
\title{Offline and Online Models of Budget Allocation for\\Maximizing Influence Spread}

\author{Noa Avigdor-Elgrabli\thanks{Yahoo Labs, Haifa, Israel. Email: {\tt noaa@oath.com}. }
	\and Gideon Blocq\thanks{Technion, Haifa, Israel. Email: {\tt gideon@alumni.technion.ac.il}. This work was done when the author was an intern at Yahoo Labs, Haifa, Israel.} 
	\and Iftah Gamzu\thanks{Amazon, Israel. Email: {\tt iftah.gamzu@yahoo.com}. This work was done when the author was affiliated with Yahoo Labs, Haifa, Israel.}
	\and Ariel Orda\thanks{Technion, Haifa, Israel. Email: {\tt ariel@ee.technion.ac.il}.}
	}
\date{}
\maketitle

\begin{abstract}
The research of influence propagation in social networks via word-of-mouth processes has been given considerable attention in recent years. Arguably, the most fundamental problem in this domain is the influence maximization problem, where the goal is to identify a small seed set of individuals that can trigger a large cascade of influence in the network. While there has been significant progress regarding this problem and its variants, one basic shortcoming of the underlying models is that they lack the flexibility in the way the overall budget is allocated to different individuals. Indeed, budget allocation is a critical issue in advertising and viral marketing. Taking the other point of view, known models allowing flexible budget allocation do not take into account the influence spread in social networks. We introduce a generalized model that captures both budgets and influence propagation simultaneously. 

For the offline setting, we identify a large family of natural budget-based propagation functions that admits a tight approximation guarantee. This family extends most of the previously studied influence models, including the well-known Triggering model. We establish that any function in this family implies an instance of a monotone submodular function maximization over the integer lattice subject to a knapsack constraint. This problem is known to admit an optimal $1-1/e \approx 0.632$-approximation. We also study the price of anarchy of the multi-player game that extends the model and establish tight results. 

For the online setting, in which an unknown subset of agents arrive in a random order and the algorithm needs to make an irrevocable budget allocation in each step, we develop a $1/(15e) \approx 0.025$-competitive algorithm. This setting extends the celebrated secretary problem, and its variant, the submodular knapsack secretary problem. Notably, our algorithm improves over the best known approximation for the latter problem, even though it applies to a more general setting.
\end{abstract}	

\thispagestyle{empty}
\end{titlepage}

\section{Introduction}
The study of information and influence propagation in societies has received increasing attention for several decades in various areas of research. Recently, the emergence of online social networks brought forward many new questions and challenges regarding the dynamics by which information, ideas, and influence spread among individuals. One central algorithmic problem in this domain is the \textit{influence maximization} problem, where the goal is to identify a small seed set of individuals that can trigger a large word-of-mouth cascade of influence in the network. This problem has been posed by \name{Domingos and Richardson}~\cite{DomingosR01,RichardsonD02} in the context of viral marketing. The premise of viral marketing is that by targeting a few influential individuals as initial adopters of a new product, it is possible to trigger a cascade of influence in a social network. Specifically, those individuals are assumed to recommend the product to their friends, who in turn recommend it to their friends, and so on.

The influence maximization problem was formally defined by \name{Kempe, Kleinberg and Tardos}~\cite{KempeKT03,KempeKT05}. In this setting, we are given a social network graph, which represents the individuals and the relationships between them. We are also given an influence function that captures the expected number of individuals that become influenced for any given subset of initial adopters. Given some budget $b$, the objective is to find a seed set of $b$ initial adopters that will maximize the expected number of influenced individuals. Kempe et al.\ studied several operational models representing the step-by-step dynamics of propagation in the network, and analyzed the influence functions that are derived from them. While there has been significant progress regarding those models and related algorithmic issues, one shortcoming that essentially has not been treated is the lack of flexibility in the way that the budget is allocated to the individuals. Indeed, budget allocation is a critical factor in advertising and viral marketing. This raises some concerns regarding the applicability of current techniques.

Consider the following scenario as a motivating example. A new daily deals website is interested in increasing its exposure to new audience. Consequently, it decides to provide discounts to individuals who are willing to announce their purchases in a social network. The company has several different levels of discounts that it can provide to individuals to incentivize them to better communicate their purchases, e.g., making more enthusiastic announcements on their social network. The company hopes that those announcements will motivate the friends of the targeted individuals to visit the website, so a word-of-mouth process will be created. The key algorithmic question for this company is whom should they offer a discount, and what amount of discounts should be offered to each individual.

\name{Alon et al.}~\cite{AlonGT12} have recently identified the insufficiency of existing models to deal with budgets. They introduced several new models that capture issues related to budget distribution among potential influencers in a social network. One main caveat in their models is that they do not take into account the influence propagation that happens in the network. The main aspect of our work targets this issue.

\subsection{Our results}
We introduce a generalized model that captures both budgets and influence propagation simultaneously. Our model combines design decisions taken from both the budget models~\cite{AlonGT12} and propagation models~\cite{KempeKT03}. The model interprets the budgeted influence propagation as a two-stage process consisting of: (1) influence related directly to the budget allocation, in which the seed set of targeted individuals influence their friends based on their budget, and (2) influence resulting from a secondary word-of-mouth process in the network, in which no budgets are involved. Note that the two stages give rise to two influence functions whose combination designates the overall influence function of the model. We study our model in both offline and online settings. 

\smallskip \noindent {\bf An offline setting.} We identify a large family of natural budget-based propagation functions that admits a tight approximation guarantee. Specifically, we establish sufficient properties for the two influence functions mentioned above, which lead to a resulting influence function that is both monotone and submodular. It is important to emphasize that the submodularity of the combined function is not the classical set-submodularity, but rather, a generalized version of submodularity over the integer lattice. Crucially, when our model is associated with such an influence function, it can be interpreted as a special instance of a {\em monotone submodular function maximization over the integer lattice subject to a knapsack constraint}. This problem is known to have an efficient algorithm whose approximation ratio of $1-1/e \approx 0.632$, which is best possible under P $\neq$ NP assumption ~\cite{SomaKIK14}. We then focus on social networks scenario, and introduce a natural budget-based influence propagation model that we name {\em Budgeted Triggering}. This model extends many of the previously studied influence models in networks. Most notably, it extends the well-known Triggering model~\cite{KempeKT03}, which in itself generalizes several models such as the Independent Cascade and the Linear Threshold models. We analyze this model within the two-stage framework mentioned above, and demonstrate that its underlying influence function is monotone and submodular. Consequently, we can approximate this model to within a factor of $1-1/e$. We also consider a multi-player game that extends our model. In this game, there are multiple players, each of which is interested to spend her budget in a way that maximizes her own network influence. We establish that the price of anarchy (PoA) of this game is equal to $2$. This result is derived by extending the definition of a monotone utility game on the integer lattice~\cite{MaeharaYK15}. Specifically, we show that one of the conditions of the utility game can be relaxed, while still maintaining a PoA of at most $2$, and that the refined definition captures our budgeted influence model.

\smallskip \noindent {\bf An online setting.} In the online setting, there is unknown subset of individuals that arrive in a random order. Whenever an individual arrives, the algorithm learns the marginal influence for each possible budget assignment, and needs to make an irrevocable decision regarding the allocation to that individual. This allocation cannot be altered later on. Intuitively, this setting captures the case in which there is an unknown influence function that is partially revealed with each arriving individual. Similarly to before, we focus on the case that the influence function is monotone and submodular. Note that this setting is highly-motivated in practice. As observed by \name{Seeman and Singer~}\cite{SeemanS13}, in many cases of interest, online merchants can only apply marketing techniques on individuals who have engaged with them in some way, e.g., visited their online store. This gives rise to a setting in which only a small unknown sample of individuals from a social network arrive in an online fashion. We identify that this setting generalizes the \textit{submodular knapsack secretary} problem~\cite{BateniHZ13}, which in turn, extends the well-known \textit{secretary} problem~\cite{Dynkin63}. We develop a $1 /(15e) \approx 0.025$-competitive algorithm for the problem. Importantly, our results not only apply to a more general setting, but also improve the best known competitive bound for the former problem, which is $1 /(20e) \approx 0.018$, due to \name{Feldman, Naor and Schwartz~}\cite{FeldmanNS11}. 
\ifthenelse{\equal{\Context}{ABSTRACT}}
{Due to space limitations, some proofs are omitted from this extended abstract.} 
{} 

\subsection{Related work}
Models of influence spread in networks are well-studied in social science~\cite{Granovetter1978} and marketing literature~\cite{GoldenbergLM2001}. \name{Domingos and Richardson~}\cite{DomingosR01,RichardsonD02} were the first to pose the question of finding influential individuals who will maximize adoption
through a word-of-mouth effect in a social network. \name{Kempe, Kleinberg and Tardos~}\cite{KempeKT03,KempeKT05} formally modeled this question, and proved that several important models have submodular influence functions. Subsequent research have studied extended models and their characteristics~\cite{Kleinberg07,MosselR07,BharathiKS07,HartlineMS08,Chen09,Singer12,EftekharGK13,SeemanS13,KhannaL14,DemaineHMMRSZ14}, and developed techniques for inferring influence models from observable data~\cite{GoyalBL10,Gomez-RodriguezLK12,LeiMMCS15}. Influence maximization with multiple players has also been considered in the past~\cite{BharathiKS07,GoyalK12,HeK13}. Kempe et al.~\cite{KempeKT03}, and very recently, Yang et al.~\cite{YangMPH16}, studied propagation models that have a similar flavor to our budgeted setting. We like to emphasize that there are several important distinctions between their models and ours. Most importantly, their models assume a strong type of \textit{fractional} diminishing returns property that our \textit{integral} model does not need to satisfy. Therefore, their models cannot capture the scenarios we describe. The reader may refer to the cited papers and the references therein for a broader review of the literature. 

\name{Alon et al.~}\cite{AlonGT12} studied models of budget allocation in social networks. As already mentioned, our model follows some of the design decisions in their approach. For example, their models support constraints on the amount of budget that can be assigned to any individual. Such constraints are motivated by practical marketing conditions set by policy makers and regulations. Furthermore, their models focus on a discrete integral notion of a budget, which is consistent with common practices in many organizations (e.g., working in multiplications of some fixed value) and related simulations~\cite{ShihL2001}. We include those considerations in our model as well. Alon et al.\ proved that one of their models, namely, the budget allocation over bipartite influence model, admits an efficient $(1-1/e)$-approximation algorithm. This result was extended by \name{Soma et al.~}\cite{SomaKIK14} to the problem of maximizing a monotone submodular function over the integer lattice subject to a knapsack constraint. The algorithm for the above problems is a reminiscent of the algorithm for maximizing a monotone submodular set function subject to a knapsack constraint~\cite{Sviridenko04}. Note that none of those papers have taken into consideration the secondary propagation process that occurs in social networks. 

The classical secretary problem was introduced more than 50 years ago (e.g.,~\cite{Dynkin63}). Since its introduction, many variants and extension of that problem have been proposed and analyzed~\cite{Kleinberg05,BabaioffIKK07,BabaioffIKK08,BarmanUCM12}. The problem that is closest to the problem implied from our online model is the submodular knapsack secretary problem~\cite{BateniHZ13,GuptaRST10,FeldmanNS11}. An instance of this problem consists of a set of $n$ secretaries that arrive in a random order, each of which has some intrinsic cost. An additional ingredient of the input is a monotone submodular set function that quantifies the value gained from any subset of secretaries. The objective is to select a set of secretaries of maximum value under the constraint that their overall cost is no more than a given budget parameter. Note that our model extends this setting by having a more general influence function that is submodular over the integer lattice. Essentially, this adds another layer of complexity to the problem as we are not only required to decide which secretaries to select, but we also need to assign them budgets.


\section {Preliminaries}\label{sec:prelim}
We begin by introducing a very general budgeted influence propagation model. This model will be specialized later when we consider the offline and online settings. In our model, there is a set of $n$ agents and an influence function $f: \mathbb{N}^n \rightarrow \mathbb{R}_+$. Furthermore, there 
is a capacity vector $c \in \mathbb{N}_+^n$ and a budget $B \in \mathbb{N}_+$. Our objective is to compute a budget assignment to the agents $b \in \mathbb{N}^n$, which maximizes the influence $f(b)$. The vector $b$ must (1) respect the capacities, that is, $0 \leq b_i \leq c_i$, for every $i \in [n]$, (2) respect the total budget, namely, $\sum_{i=1}^n b_i \leq B$. In the following, we assume without loss of generality that each $c_i \leq B$.

We primarily focus on influence functions that maintain the properties of monotonicity and submodularity. A function $f: \mathbb{N}^n \rightarrow \mathbb{R}_+$ is called \textit{monotone} if $f(x) \leq f(y)$ whenever $x \leq y$ coordinate-wise, i.e., $x_i \le y_i$, for every $i \in [n]$. The definition of submodularity for functions over the integer lattice is a natural extension of the classical definition of submodularity over sets (or boolean vectors):

\begin{definition}
A function $f: \mathbb{N}^n \rightarrow \mathbb{R}_+$ is said to be \emph{submodular over the integer lattice} if $f(x) + f(y) \geq f(x \vee y) + f(x \wedge y)$, for all integer vectors $x$ and $y$, where $x \vee y$ and $x \wedge y$ denote the coordinate-wise maxima and minima, respectively. Specifically, $(x \vee y)_i = \max\{x_i,y_i\}$ and $(x \wedge y)_i = \min\{x_i,y_i\}$.
\end{definition}

In the remainder of the paper, we abuse the term submodular to describe both set functions and functions over the integer lattice. We also make the standard assumption of a \textit{value oracle} access for the function $f$. A value oracle for $f$ allows us to query about $f(x)$, for any vector $x$. The question of how to compute the function $f$ in an efficient (and approximate) way has spawned a large body of work in the context of social networks (e.g.,~\cite{KempeKT03,ChenWY09,ChenWW10,MathioudakisBCGU11,CohenDPW14,BorgsBCL14}). 

Notice that for the classical case of sets, the submodularity condition implies that $f(S) + f(T) \ge f(S\cup T) + f(S\cap T)$, for every $S, T \subseteq [n]$, and the monotonicity property implies that $f(S) \leq f(T)$ if $S \subseteq T$. An important distinction between the classical set setting and the integer lattice setting can be seen when we consider the \textit{diminishing marginal returns} property. This property is an equivalent definition of submodularity of set functions, stating that  $f(S \cup \{i\}) - f(S) \ge f(T \cup \{i\})-f(T) $, for every $S \subseteq T$ and every $i \notin T$. However, this property, or more accurately, its natural extension, does not characterize submodularity over the integer lattice, as observed by \name{Soma et al.~}\cite{SomaKIK14}. For example, there are simple examples of a submodular function $f$ for which 
$$
f(x + \chi_i) - f(x) \geq f(x + 2\chi_i) - f(x + \chi_i)
$$
does not hold. Here, $\chi_i$ is the characteristic vector of the set $\{i\}$, so that $x + k\chi_i$ corresponds to an update of $x$ by adding an integral budget $k$ to agent $i$. Note that a weaker variant of the diminishing marginal returns does hold for submodular functions over the integer lattice. 
\begin{lemma}[Lem~2.2 \cite{SomaKIK14}] \label{lm:submod_property_1}
Let $f$ be a monotone submodular function over the integer lattice. For any $i \in [n]$, $k \in \mathbb{N}$, and $x \leq y$, it follows that 
$$
f(x \vee k \chi_i) - f(x) \geq f(y \vee k\chi_i) - f(y) \ .
$$
\end{lemma}

\section{An Offline Model}	\label{sec:offline}
In this section, we study the offline version of the budgeted influence propagation model. As already noted, we consider budgeted influence propagation to be a two-stage process consisting of (1) direct influence related to the budget assignment, followed by (2) influence related to a propagation process in the network. In particular, in the first stage, the amount of budget allocated to an individual determines her level of effort and success in influencing her direct friends. This natural assumption is consistent with previous work~\cite{AlonGT12}. Then, in the second stage, a word-of-mouth propagation process takes place in which additional individuals in the network may become affected. Note that the allocated budgets do not play role at this stage. 

We identify a large family of budget-based propagation functions that admit an efficient solution. Specifically, we first identify sufficient properties of the influence functions of both stages, which give rise to a resulting (combined) influence function that is monotone and submodular. Consequently, our model can be interpreted as an instance of a monotone submodular function maximization over the integer lattice subject to a knapsack constraint. This problem is known to have an efficient $(1-1/e)$-approximation~\cite{SomaKIK14}, which is best possible under the assumption that P $\neq$ NP. This NP-hardness bound of $1-1/e$ already holds for the special case of maximum coverage~\cite{Feige98,AlonGT12}.

We subsequently develop a natural model of budgeted influence propagation in social networks that we name \textit{Budgeted Triggering}. This model generalizes many settings, including the well-known Triggering model. Note that the Triggering model already extends several models used to capture the spread of influence in networks, like the Independent Cascade, Linear Threshold, and Listen Once models~\cite{KempeKT03}. We demonstrate that the influence function defined by this model is monotone and submodular, and thus, admits an efficient $(1-1/e)$-approximation. Technically, we achieve this result by demonstrating that the two-stage influence functions that underlie this model satisfy the sufficient properties mentioned above. 

Finally, we study an extension of the Budgeted Triggering model to a multi-player game. In this game, there are multiple self-interested players (e.g., advertisers), each of which is interested to spend her budget in a way that maximizes her own network influence. We establish that the price of anarchy (PoA) of the game is exactly $2$. In fact, we prove that this result holds for a much more general type of games. Maehara et al.~\cite{MaeharaYK15} recently defined the notion of a monotone utility game on the integer lattice, and demonstrated that its PoA is at most $2$. Their utility game definition does not capture our multi-player game. We show that one of the conditions in their game definition can be relaxed while still maintaining the same PoA. Crucially, this relaxed definition captures our model.

\subsection{A two-stage influence composition}
The two-stage process can be formally interpreted as a composition of two influence functions, $f = h \circ g$. The first function 
$g:\mathbb{N}^n \rightarrow \{0,1\}^n$ captures the set of influenced agents for a given budget allocation, while the second function $h: \{0,1\}^n \rightarrow \mathbb{R}_+$ captures the overall number (or value) of influenced agents, given some seed agent set for a propagation process. In particular, the influenced agents of the first stage are the seed set for the second stage. We next describe sufficient conditions for the functions $g$ and $h$ which guarantee that their composition is monotone and submodular over the integer lattice. Note that we henceforth use notation related to sets and their binary vector representation interchangeably.
\begin{definition}
A function $g: \mathbb{N}^n \rightarrow \{0,1\}^n$ is said to be \emph{coordinate independent} if it satisfies $g(x \vee y) \leq g(x) \vee g(y)$, for any $x,y \in \mathbb{N}^n$.
\end{definition}

\begin{definition}
A function $g: \mathbb{N}^n \rightarrow \{0,1\}^n$ is said to be \emph{monotone} if $g(x) \leq g(y)$ coordinate-wise whenever $x\leq y$ coordinate-wise.
\end{definition}

Many natural influence functions are coordinate independent. One such example is the family of functions in which the output vector is a coordinate-wise disjunction over a set of $n$ vectors, each of which captures the independent influence implied by some agent. Specifically, the $i$th vector in the disjunction is the result of some function $f_i: \mathbb{N} \rightarrow \{0,1\}^n$ indicating the affected agents as a result of any budget allocation assigned only to agent $i$. We are now ready to prove our composition lemma. 

\begin{lemma} \label{lem:compose}
Given a monotone coordinate independent function $g: \mathbb{N}^n \rightarrow \{0,1\}^n$ and a monotone submodular function $h: \{0,1\}^n \rightarrow \mathbb{R}_+$, the composition $f = h \circ g: \mathbb{N}^n \rightarrow \mathbb{R}_+$ is a monotone submodular function over the integer lattice.	
\end{lemma}
\ifthenelse{\equal{\Context}{ABSTRACT}}
{} 
{
\begin{proof}
The coordinate independence properties of $g$ and the monotonicity of $h$ imply that 
$$
h(g(x) \vee g(y))  \geq  h(g(x\vee y)).
$$
In addition, from the monotonicity of $g$ we know that $g(x \wedge y) \leq g(x)$ and $g(x \wedge y) \leq g(y)$. Thus, together with the monotonicity of $h$, we get that
$$
h(g(x) \wedge g(y)) \geq  h(g(x \wedge y)). 
$$
Utilizing the above results, we attain that $f$ is submodular since 
\begin{eqnarray*}
	f(x) + f(y)& = & h(g(x)) + h(g(y))\\
	& \geq & h(g(x) \vee g(y)) + h(g(x) \wedge g(y))\\ 
	&\geq& h(g(x \vee y)) + h(g(x \wedge y))\\
	& = & f(x \vee y) + f (x \wedge y) \ , 
\end{eqnarray*}
where the first inequality is by the submodularity of $h$. 

We complete the proof by noting that $f$ is monotone since both $g$ and $h$ are monotone. Formally, given $x \leq y$ then it follows that $f(x) = h(g(x)) \leq h(g(y)) = f(y)$ by $h$'s monotonicity and since $g(x) \leq g(y)$ by $g$'s monotonicity.~
\end{proof}
} 

As a corollary of the lemma, we get the following theorem. 
\begin{theorem} \label{th:compose}
Given a monotone coordinate independent function $g: \mathbb{N}^n \rightarrow \{0,1\}^n$ and a monotone submodular function $h: \{0,1\}^n \rightarrow \mathbb{R}_+$, there is a $(1-1/e)$-approximation algorithm for maximizing the influence function $f = h \circ g: \mathbb{N}^n \rightarrow \mathbb{R}_+$ under capacity constraints $c \in \mathbb{N}_+^n$ and a budget constraint $B \in \mathbb{N}_+$, whose running time is polynomial in $n$, $B$, and the query time of the value oracle for $f$.
\end{theorem}
\ifthenelse{\equal{\Context}{ABSTRACT}}
{} 
{
\begin{proof}
We know by Lemma~\ref{lem:compose} that $f$ is monotone and submodular. Consequently, we attain an instance of maximizing a monotone submodular function over the integer lattice subject to a knapsack constraint. \name{Soma et al.~}\cite{SomaKIK14} recently studied this problem, and developed a $(1-1/e)$-approximation algorithm whose running time is polynomial in $n$, $B$, and the query time for the value oracle of the submodular function.~
\end{proof}
} 

\subsection{The budgeted triggering model} 
We now focus on social networks, and introduce a natural budget-based influence model that we call the Budgeted Triggering model. This model consists of a social network, represented by a directed graph $G = (V,E)$ with $n$ nodes (agents) and a set $E$ of directed edges (relationships between agents). In addition, there is a function $f: \mathbb{N}^n \rightarrow \mathbb{R}_+$ that quantifies the influence of any budget allocation  $b \in \mathbb{N}^n$ to the agents. The concrete form of $f$ strongly depends on the structure of the network, as described later. The objective is to find a budget allocation $b$ that maximizes the number of influenced nodes, while respecting the feasibility constraints: (1) $b_i \le c_i$, for every node $i \in V$, and (2) $\sum_{i \in V} b_i \le B$.

For ease of presentation, we begin by describing the classic Triggering model~\cite{KempeKT03}. Let $N(v)$ be the set of neighbors of node $v$ in the graph. The influence function implied by a Triggering model is defined by the following simple process. Every node $v \in V$ independently chooses a random \textit{triggering set} $T^v \subseteq N(v)$ among its neighbors according to some fixed distribution. Then, for any given seed set of nodes, its influence value is defined as the result of a deterministic cascade process in the network which works in steps. In the first step, only the selected seed set is affected. At step $\ell$, each node $v$ that is still not influenced becomes influenced if any of its neighbors in $T^v$ became influenced at time $\ell-1$. This process terminates after at most $n$ rounds.

Our generalized model introduces the notion of budgets into this process. Specifically, the influence function in our case adheres to the following process. Every node $v$ independently chooses a random \textit{triggering vector} ${t}^v \in \mathbb{N}^{|N(v)|}$ according to some fixed distribution. Given a budget allocation $b \in \mathbb{N}^{n}$, the influence value of that allocation is the result of the following deterministic cascade process. In the first step, every node $v$ that was allocated a budget $b_v > 0$ becomes affected. At step $\ell$, every node $v$ that is still not influenced becomes influenced if any of its neighbors $u \in N(v)$ became influenced at time $\ell - 1$ and $b_u \geq t^v_u$. One can easily verify that the Triggering model is a special case of our Budgeted Triggering model, where the capacity vector $c = 1^n$, and each $t^v_u = 0$ if $u \in T^v$, and $t^v_u = B+1$, otherwise. 

Intuitively, the triggering vectors in our model capture the amount of effort that is required from each neighbor of some agent to affect her. Of course, the underlying assumption is that the effort of individuals correlates with the budget they receive. As an example, consider the case that a node $v$ selects a triggering value $t^v_u = 1$ for some neighbor $u$. In this case, $u$ can only influence $v$ if it receives a budget of at least $1$. However, if $v$ selects a value $t^v_u = 0$ then it is enough that $u$ becomes affected in order to influence $v$. In particular, it is possible that $u$ does not get any budget but still influences $v$ after it becomes affected in the cascade process.

Given a budget allocation $b$, the value of the influence function $f(b)$ is the expected number of nodes influenced in the cascade process, where the expectation is taken over the random choices of the model. Formally, let $\sigma$ be some fixed choice of the triggering vectors of all nodes (according to the model distribution), and let $\mathrm{Pr}(\sigma)$ be the probability of this outcome. Let $f_{\sigma}(b)$ be the (deterministic) number of nodes influenced when the triggering vectors are defined by $\sigma$ and the budget allocation is $b$. Then, $f(b) = \sum_{\sigma} \mathrm{Pr}(\sigma) \cdot f_\sigma(b)$.

\begin{theorem} \label{thm:submod_proof}
There is a $(1-1/e)$-approximation algorithm for influence maximization under the Budgeted Triggering model whose running time is polynomial in $n$, $B$, and the query time of the value oracle for the influence function.
\end{theorem}
\ifthenelse{\equal{\Context}{ABSTRACT}}
{} 
{
\begin{proof}
Consider an influence function $f: \mathbb{N}^n \rightarrow \mathbb{R}_+$ resulting from the Budgeted Triggering model. We next show that the function $f$ is monotone and submodular over the integer lattice. As a result, our model can be interpreted as an instance of maximizing a monotone submodular function over the integer lattice subject to a knapsack constraint, which admits an efficient $(1-1/e)$-approximation. Notice that it is sufficient to prove that each (deterministic) function $f_\sigma$ is monotone submodular function over the integer lattice. This follows as $f$ is a non-negative linear combination of all $f_\sigma$. One can easily validate that submodularity and monotonicity are closed under non-negative linear combinations. 

Consider some function $f_\sigma: \mathbb{N}^n \rightarrow \mathbb{R}_+$. For the purpose of establishing that $f_\sigma$ is monotone and submodular, we show that $f_\sigma$ can be interpreted as a combination of a monotone coordinate independent function $g_\sigma: \mathbb{N}^n \rightarrow \{0,1\}^n$, and a monotone submodular function $h_\sigma: \{0,1\}^n \rightarrow \mathbb{R}_+$. The theorem then follows by utilizing Lemma~\ref{lem:compose}. We divide the diffusion process into two stages. In the first stage, we consider the function $g_\sigma$, which given a budget allocation returns (the characteristic vector of) the set $S$ of all the nodes that were allocated a positive budget along with their immediate neighbors that were influenced according to the Budgeted Triggering model. Formally,
$$
g_{\sigma}(b) = S \triangleq \big\{v : b_v > 0 \big\} \cup \big\{u : \exists v \in N(u), b_v > 0, b_v \geq t^u_v \big\} \ .
$$
In the second stage, we consider the function $h_{\sigma}$ that receives (the characteristic vector of) $S$ as its seed set, and makes the (original) Triggering model interpretation of the vectors. Specifically, the triggering set of each node $v$ is considered to be $T^v = \{u : t^v_u = 0\}$. Intuitively, the function $g_\sigma$ captures the initial budget allocation step and the first step of the propagation process, while the function $h_\sigma$ captures all the remaining steps of the propagation. Observe that $f_\sigma = h_\sigma \circ g_\sigma$ by our construction. Also notice that $h_{\sigma}$ is trivially monotone and submodular as it is the result of a Triggering model~\cite[Thm.~4.2]{KempeKT03}. Therefore, we are left to analyze the function $g_\sigma$, and prove that it is monotone and coordinate independent. The next claim establishes these properties, and completes the proof of the theorem.
\end{proof}

\begin{claim} \label{clm:submod_sigma}
The function $g_\sigma$ is monotone and coordinate independent.
\end{claim}
\begin{proof}
Let $x,y \in \mathbb{N}^n$, and denote $w = x \vee y$. We establish coordinate independence by considering each influenced node in $g_\sigma(w)$ separately. 
Recall that $g_\sigma(w)$ consist of the union of two sets $\{v : w_v>0\}$ and $\{u : \exists v \in N(u), w_v >0,  w_v \geq t^u_v \}$.
Consider a node $v$ for which $w_v > 0$. Since $w_v = \max\{x_v, y_v\}$, we  know that at least one of $\{x_v, y_v\}$ is equal to $w_v$, say $w_v=x_v$. Hence, $v \in g_\sigma(x)$. Now, consider a node $u \in g_\sigma(w)$ having $w_u=0$. It must be the case that $u$ is influenced by one of its neighbors $v$. Clearly, $w_v > 0$ and $w_v \ge t^u_v$. Again, we can assume without loss of generality that $w_v = x_v$, and get that $u \in g_{\sigma}(x)$. This implies that for each $v \in g_\sigma(x \vee y)$, either $v \in g_\sigma(x)$ or $v \in g_\sigma(y)$, proving coordinate independence, i.e., $g_\sigma({x} \vee {y}) \leq g_\sigma({x}) \vee  g_\sigma({y})$.

We prove monotonicity in a similar way. Let $x \leq y$. Consider a node $v \in g_\sigma(x)$ for which $x_v > 0$. Since $y_v \geq x_v > 0$, we know that $v\in g_\sigma(y)$. Now, consider a node $u \in g_\sigma(x)$ having $x_u = 0$. There must be a node $v \in N(u)$ such that $x_v >0$, and $x_v \ge t^u_v$. Accordingly, we get that $y_v \ge t^u_v$, and hence, $u \in g_\sigma(y)$. This implies that $g_\sigma(x) \leq g_\sigma(y)$, which completes the proof.~
\end{proof}
} 

\subsection{A multi-player budgeted influence game}
We now focus on a multi-player budgeted influence game. In the general setting of the game, which is formally defined by the tuple $(M, (A^i)^M_{i=1}, (f^i)^M_{i=1})$, there are $M$ self-interested players, each of which needs to decide how to allocate its budget $B^i \in \mathbb{N}_+$ among $n$ agents. Each player has a capacity vector $c^i \in \mathbb{N}_{+}^n$ that bounds the amount of budget she may allocate to every agent. The budget assignment of player $i$ is denoted by $b^i \in \mathbb{N}_{+}^n$, and is referred to as its strategy. The strategy of player $i$ is feasible if it respects the constraints: (1) $b^i_j \leq c^i_j$, for every $j \in [n]$, and (2) $\sum_{j=1}^n b^i_j \leq B^i$. Let $A^i$ be the set of all feasible strategies for player $i$. Note that we allow mixed (randomized) strategies. Each player has an influence function $f^i: \mathbb{N}^{M \times n} \rightarrow \mathbb{R}_+$ that designates her own influence (payoff) in the game. Specifically, $f^i(b)$ is the payoff of player $i$ for the budget allocation $b$ of all players. This can also be written as $f^i(b^i, b^{-i})$, where the strategy of $i$ is $b^i$ and the strategies of all the other players are marked as $b^{-i}$. Note that the goal of each player is to maximize her own influence, given her feasibility constraints and the strategies of other players. Let $F(b) = \sum_{i=1}^M f^i(b)$ be the social utility of all players in the game. 

One of the most commonly used notions in game theory is \textit{Nash equilibrium} (NE)~\cite{Nash50}. This notion translates to our game as follows: A budget allocation $b$ is said to be in a NE if $f^i(b^i, b^{-i}) \geq f^i(\tilde{b}^i, b^{-i})$, for every $i$ and $\tilde{b}^i \in A^i$.

\smallskip \noindent {\bf Monotone utility game on the integer lattice.}
We begin by studying a monotone utility game on the integer lattice, and establish that its PoA is no more than $2$. Later on, we demonstrate that our Budgeted Triggering model can be captured by this game. Utility games were defined for submodular set functions by Vetta~\cite{Vetta02}, and later extended to submodular functions on the integer lattice by Maehara et al.~\cite{MaeharaYK15}. We build on the latter work, and demonstrate that one of the conditions in their utility game definition, namely, the requirement that the submodular function satisfies component-wise concavity, can be neglected. Note that component-wise concavity corresponds to the diminishing marginal returns property, which does not characterize submodularity over the integer lattice, as noted in Section~\ref{sec:prelim}. Therefore, removing this constraint is essential for proving results for our model.
 
We refine the definition of a monotone utility game on the integer lattice~\cite{MaeharaYK15}, so it only satisfies the following conditions:
\begin{list}{\hspace{0.6cm}(U\theenumi)}{\usecounter{enumi}}
\itemsep0em
\item $F(b)$ is a monotone submodular function on the integer lattice.
\item $F(b) \geq \sum_{i = 1}^M f^i(b)$.
\item $f^i(b) \geq F(b^i,b^{-i}) - F(0,b^{-i})$, for every $i \in [M]$.
\end{list}

\begin{theorem} \label{th:utilitygame}
The price of anarchy of the monotone utility game designated by U1-U3 is at most 2.
\end{theorem}
\begin{proof}
Let $b_* = (b_*^{1}, \ldots, b_*^{M})$ be the social optimal budget allocation, and let $b = (b^1, \ldots, b^M)$ be a budget allocation in Nash equilibrium. Let $\tilde{b}^{i} = (b_*^{1}, \ldots, b_*^{i}, 0, \ldots, 0)$ be the optimal budget allocation restricted to the first $i$ players. Notice that
\begin{eqnarray*}
F(b^*) - F(b) & \leq & F(b^* \vee b) - F(b) \\
& = & \sum_{i =1}^M F(\tilde{b}^{i} \vee b) - F(\tilde{b}^{i-1} \vee b) \\
& \leq & \sum_{i =1}^M F(b_*^{i} \vee b^i, b^{-i}) - F(b^i, b^{-i}) \ ,
\end{eqnarray*}
where the first inequality is due to the monotonicity of $F$, the equality holds by a telescoping sum, and the last inequality is due to submodularity of $F$. Specifically, submodularity implies that inequality since
$$
F(b_*^{i} \vee b^i, b^{-i}) + F(\tilde{b}^{i-1} \vee b) \geq F(\tilde{b}^{i} \vee b) +  F(b^i, b^{-i}) \ .
$$
Now, observe that
$$
F(b^i, b^{-i}) + F(b_*^{i}, b^{-i}) \geq F(b_*^{i} \vee b^i, b^{-i}) + F(b_*^{i} \wedge b^i, b^{-i}) \geq F(b_*^{i} \vee b^i, b^{-i}) + F(0, b^{-i}) \ .
$$
Here, the first inequality holds by the submodularity of $F$, while the last inequality follows from the monotonicity of $F$. Consequently, we derive that 
$$F(b^*) - F(b) \leq \sum_{i = 1}^M F(b_*^{i}, b^{-i}) - F(0, b^{-i})
\leq \sum_{i = 1}^M  f^i(b_*^{i}, b^{-i}) \leq \sum_{i = 1}^M  f^i(b^i, b^{-i}) \leq F(b) \ ,
$$
where the second inequality is by condition (U3) of the utility game, the third inequality holds since $b$ is a Nash equilibrium, and the last inequality is by condition (U2) of the utility game. This completes the proof as $F(b^*) \leq 2F(b)$.~
\end{proof}

\smallskip \noindent {\bf A multi-player Budgeted Triggering model.}
We extend the Budgeted Triggering model to a multi-player setting. As before, we have a social network, represented by a directed graph $G = (V,E)$, such that every node $v$ has an independent random triggering vector ${t}^v \in \mathbb{N}^{|N(v)|}$. Each player $i$ has a budget $B^i \in \bbN_+$, and a function $f^i: \mathbb{N}^{M \times n} \rightarrow \mathbb{R}_+$ that quantifies her influence, given the budget allocation of all players. The objective of each player $i$ is to find a budget allocation $b^i$, given the budget allocations of other players, that maximizes the number of nodes that she influences, while respecting the feasibility constraints: (1) $b^i_j \le c^i_j$, for every node $j \in V$, and (2) $\sum_{j \in V} b^i_j \le B$.

The process in which nodes become affected is very similar to that in Budgeted Triggering, but needs some refinement for the multi-player setting. We follow most design decisions of Bharathi et al.~\cite{BharathiKS07}. Specifically, whenever a player influences a node, this node is assigned the color of that player. Once a node become influenced, its color cannot change anymore. If two or more players provide positive budgets to the same node, then the node is given the color of the player that provided the highest budget. In case there are several such players, the node is assigned a color uniformly at random among the set of players with the highest budget assignment. If a node $u$ becomes influenced at step $\ell$, it attempts to influence each of its neighbors. If the activation attempt from $u$ to its neighbor $v$ succeeds, which is based on the triggering vector of $v$, then $v$ becomes influenced with the same color as $u$ at step $\ell + T_{uv}$, assuming that it has not been influenced yet. All $T_{uv}$'s are independent positive continuous random variables. This essentially prevents simultaneous activation attempts by multiple neighbors.

\begin{lemma}\label{lem:social_submod}
The social function $F(b) = \sum_{i = 1}^M f^i(b)$ is a monotone submodular function on the integer lattice.
\end{lemma}
\ifthenelse{\equal{\Context}{ABSTRACT}}
{} 
{
\begin{proof}
We prove this lemma along similar lines to those in the proof of Theorem~\ref{thm:submod_proof}, which attains to the single-player scenario. Let $\sigma$ be some fixed choice of triggering vectors of all nodes and all the activation times $T_{uv}$. We also assume that $\sigma$ encodes other random decisions in the model, namely, all tie-breaking choices related to equal (highest) budget assignments for nodes. Let $F_{\sigma}(b)$ be the deterministic number of influenced nodes for the random choices $\sigma$ and the budget allocation $b$, and note that
$F(b) = \sum_{\sigma} \mathrm{Pr}(\sigma) \cdot F_\sigma(b)$.
Similar to Theorem \ref{thm:submod_proof}, it is sufficient to prove that $F_{\sigma}$ is monotone submodular function on the integer lattice. Again, we view the social influence as a two-stage process. In the first step, we consider a function $G_\sigma: \bbN^{M \times n} \to \{0,1\}^n$ that given the budget allocation of all players returns a set $S$ of immediate influenced nodes. Formally,
$$
G_{\sigma}(b) = S \triangleq \big\{v : \exists i,~b^i_v > 0 \big\} \cup 
\big\{u : \exists v \in N(u), \exists i,~b^i_v > 0, b^i_v \geq t^u_v \big\} \ .
$$
In the second stage, we consider the function $H_{\sigma}$ that receives $S$ as its seed set, and makes the original Triggering model interpretation of the vectors, that is, it sets each $T^v = \{u : t_u^v = 0\}$. Notice that the fact that there are multiple players at this stage does not change the social outcome, i.e., the number of influenced nodes, comparing to a single-player scenario. The only difference relates to the identity of the player that affects every node. This implies that $H_{\sigma}$ is monotone and submodular as its result is identical to that of the original (single-player) Triggering model~\cite{KempeKT03}. Observe that $F_\sigma = H_\sigma \circ G_\sigma$ by our construction. Therefore, by Lemma~\ref{lem:compose}, we are left to establish that the function $G_\sigma$ is monotone and coordinate independent. The next claim proves that.

\begin{claim}
The function $G_\sigma$ is monotone and coordinate independent.
\end{claim}
We prove this claim using almost identical line of argumentation to that in Claim~\ref{clm:submod_sigma}. Let $x,y \in \mathbb{N}^{M \times n}$, and denote $w = x \vee y$. We establish coordinate independence by considering every affected node in $G_\sigma(w)$ separately. Recall that $G_\sigma(w)$ consist of the union of two sets $\{v : \exists i,~w^i_v > 0\}$ and $\{u : \exists v \in N(u), \exists i,~w^i_v > 0, w^i_v \geq t^u_v \}$. Consider a node $v$ that has some player $i$ with $w^i_v > 0$. Since $w^i_v = \max\{x^i_v, y^i_v\}$, we know that at least one of $\{x^i_v, y^i_v\}$ is equal to $w^i_v$, say $w^i_v=x^i_v$. Hence, $v \in G_\sigma(x)$, since in particular, player $i$ competes on influencing $v$. Now, consider a node $u \in G_\sigma(w)$ with $w^i_u=0$, for all $i$. It must be the case that $u$ is influenced by one of its neighbors $v$. Clearly, there exists some player $i$ such that $w^i_v > 0$ and $w^i_v \ge t^u_v$. Again, we can assume without loss of generality that $w^i_v = x^i_v$, and get that $u \in G_{\sigma}(x)$, since in particular, player $i$ competes on influencing $u$ via $v$. This implies that for each $v \in G_\sigma(x \vee y)$, either $v \in G_\sigma(x)$ or $v \in G_\sigma(y)$, proving coordinate independence, i.e., $G_\sigma({x} \vee {y}) \leq G_\sigma({x}) \vee  G_\sigma({y})$.

We prove monotonicity in a similar way. Let $x \leq y$. Consider a node $v \in G_\sigma(x)$ such that there is a player $i$ for which $x^i_v > 0$. Since $y^i_v \geq x^i_v > 0$, we know that player $i$ competes on influencing $v$, and thus, $v \in G_\sigma(y)$. Now, consider a node $u \in G_\sigma(x)$ with $x^i_u = 0$, for all $i$. There must be a node $v \in N(u)$ and a player $i$ such that $x^i_v >0$ and $x^i_v \ge t^u_v$. Accordingly, we get that $y^i_v \ge t^u_v$. Therefore, player $i$ also competes on influencing $u$ via $v$, and thus, $u \in G_\sigma(y)$. This implies that $G_\sigma(x) \leq G_\sigma(y)$, which completes the proof.~
\end{proof}
} 

\begin{theorem}
The Budgeted Triggering model with multiple players has a PoA of exactly $2$.
\end{theorem}
\begin{proof}
We begin by demonstrating that the model satisfies conditions U1-U3 of the monotone utility game on the integer lattice. As a result, we can apply Theorem~\ref{th:utilitygame} to attain an upper bound of $2$ on the PoA of the model. Notice that condition (U1) holds by Lemma~\ref{lem:social_submod}. Also, condition (U2) trivially holds by the definition of the social function $F$. 

For the purpose of proving that the model satisfies condition (U3), let $\sigma$ be some fixed choice of triggering vectors of all nodes and all the activation times $T_{uv}$. We also assume that $\sigma$ encodes other random decisions in the model, namely, all tie-breaking choices related to equal (highest) budget assignments for nodes. Let $F_{\sigma}(b)$ be the deterministic number of influenced nodes for the random choices $\sigma$ and the budget allocation $b$. Finally, let $f^i_{\sigma}(b)$ be the deterministic number of nodes influenced by player $i$ for the random choices $\sigma$ and the budget allocation $b$. We next argue that $f^i_{\sigma}(b) \geq F_{\sigma}(b^i,b^{-i}) - F_{\sigma}(0,b^{-i})$, for any $\sigma$. Notice that this implies condition (U3) since 
$$
f^i(b) = \sum_{\sigma} \mathrm{Pr}(\sigma)f^i_{\sigma}(b)
\geq \sum_{\sigma} \mathrm{Pr}(\sigma) \left[F_{\sigma}(b^i,b^{-i}) - F_{\sigma}(0,b^{-i})\right] 
= F(b^i,b^{-i}) - F(0,b^{-i}) \ .
$$
We turn to prove the above argument. Notice that it is sufficient to focus only on cases that $F_{\sigma}(b^i,b^{-i}) > F_{\sigma}(0,b^{-i})$, since otherwise, the argument is trivially true as $f^i_{\sigma}(b) \geq 0$.
We concentrate on all nodes $u$ that are not influenced by any player when the mutual strategy is $(0,b^{-i})$, but became influenced for a strategy $(b^i,b^{-i})$. We claim that all those nodes must be assigned the color of player $i$. It is easy to verify that increasing the budget assignment of a player to any node can only negatively affect other players, that is, they may only influence a subset of the nodes. This follows as all the activation results are deterministically encoded in the choices $\sigma$, so adding a competition can only make the outcome worse, i.e., players may not affect a node that they previously did. This implies the claim. As a result, $f^i_{\sigma}(b) \geq F_{\sigma}(b^i,b^{-i}) - F_{\sigma}(0,b^{-i})$. This completes the proof that the model is an instance of the monotone utility game on the integer lattice, and thus, has a PoA of at most $2$.

We proceed by proving the tightness of the PoA result. We show there is an instance of the multi-player Budget Triggering model whose PoA is $2N / (N+1)$. Notice that as $N \to \infty$, the lower bound on the PoA tends to $2$. This instance has been presented in a slightly different context by He and Kempe~\cite[Proposition~1]{HeK13}. Concretely, the input graph is a union of a star with one center and $N$ leaves, and $N$ additional (isolated) nodes. The triggering vectors are selected from a degenerate distribution that essentially implies that each activated node also activates all of its neighbors. Every player has one unit of budget. One can easily verify that the solution in which all players assign their unit budget to the center of the star is a NE. This follows since the expected payoff for each player is $(N+1)/N$, while unilaterally moving the budget to any other node leads to a payoff of $1$. However, the strategy that optimizes the social utility is to place one unit of budget at the center of the star graph, and the remaining budget units at different isolated nodes.~
\end{proof}

\section{An Online Model}
We study the online version of the budgeted influence propagation model. This setting can capture scenarios in which the social influences in a network are known in advance, but the (subset of) agents that will arrive and their order is unknown. The input for this setting is identical to that of the offline variant with the exception that the $n$ agents arrive in an online fashion. This intuitively means that we do not know the monotone submodular influence function $f: \mathbb{N}^n \rightarrow \mathbb{R}_{+}$ in advance, but rather, it is revealed to us gradually with time. More specifically, upon the arrival of the $i$th agent, we can infer the (constrained) function $f_i$, which quantifies the influence of $f$ for the set of the first $i$ agents, while fixing the budget of all other agents to $0$. Note that we also learn the maximum budget $c_i$ that can be allocated to agent $i$ whenever she arrives. For every arriving agent $i$, the algorithm needs to make an irrevocable decision regarding the amount of budget $b_i$ allocated to that agent without knowing the potential contribution of future arriving agents. As mentioned in the introduction, this problem is a generalization of the classical secretary problem. This immediately implies that any online algorithm preforms very poorly under an unrestricted adversarial arrival of the agents. We therefore follow the standard assumption that the agents and their influence are fixed in advanced, but their order of arrival is random.
Note that the overall influence of some budget allocation to the agents is not affected by the arrival order of the agents.

We analyze the performance of our algorithm, $\alg$, using the competitive analysis paradigm. Note that competitive analysis focuses on quantifying the cost that online algorithms suffer due to their complete lack of knowledge regarding the future, and it does not take into account computational complexity. Let $\opt$ be an optimal algorithm for the offline setting. Given an input instance $I$ for the problem, we let $\opt(I)$ and $\alg(I)$ be the influence values that $\opt$ and $\alg$ attain for $I$, respectively. We say that $\alg$ is $c$-competitive if $\inf_I \mathbb{E}[\alg(I)]/\opt(I) \geq c$, where $\mathbb{E}[\alg(I)]$ is the expected value taken over the random choices of the algorithm and the random arrival order of the agents. We like to note that our algorithm and its analysis are inspired by the results of \name{Feldman et al.~}\cite{FeldmanNS11} for the submodular knapsack secretary problem. However, we make several novel observations and identify some interesting structural properties that enable us to simultaneously generalize and improve their results. Also note that in the interests of expositional simplicity, we have not tried to optimize the constants in our analysis.

\begin{theorem}
\label{thm:online_main}
There is an online randomized algorithm that achieves $1/(15e) \approx 0.025$-competitive ratio for the budgeted influence maximization problem.
\end{theorem}
\begin{proof}
Recall that an instance of the online budgeted influence maximization problem consists of a set of $n$ agents that arrive in a random order, a budget constraint $B \in \mathbb{N}_+$, capacity constraints $c \in \mathbb{N}_+^{n}$, and a monotone submodular influence function over the integer lattice $f: \mathbb{N}^{n} \rightarrow \mathbb{R}_+$. We begin by describing the main component of our algorithm. This component is built to address the case that the contribution of each agent is relatively small with respect to the optimal solution. That is, even when one assigns the maximum feasible budget to any single agent, the contribution of that agent is still small compared to the optimum. We refer to this component as \emph{light influence} algorithm (abbreviated, LI). This component will be later complemented with another component, derived from the classical {\em secretary} algorithm, to deal with highly influential agents.

Let $\langle a_1, a_2, \ldots, a_n \rangle$ be an arbitrary fixed ordering of the set of agents. This is not necessarily the arrival order of the agents. Algorithm light influence, formally described below, assumes that each agent $a_i$ is assigned a uniform continuous random variable $t_i \in [0,1)$ that determines its arrival time. Note that this assumption does not add restrictions on the model since one can create a set of $n$ samples of the uniform distribution from the range $[0,1)$ in advance, and assign them by a non-decreasing order to each arriving agent (see, e.g., the discussion in~\cite{FeldmanNS11}). 

The algorithm begins by exploring the first part of the agent sequence, that is, the agents in $L = \{a_i : t_i \leq 1/2\}$. Note that it does not allocate any budget to those agents. Let $b^L$ be an optimal (offline) budget allocation for the restricted instance that only consists of the agents in $L$, and let $f(b^L)$ be its influence value. Furthermore, let $f(b^L) / B$ be a lower bound on the average contribution of each unit of budget in that solution. The algorithm continues by considering the remainder of the sequence. For each arriving agent, it allocates a budget of $k$ if the increase in the overall influence value is at least $\alpha k f(b^L) / B$, for some fixed $\alpha$ to be determined later. That is, the average influence contribution of an each unit of budget is (up to the $\alpha$-factor) at least as large as the average unit contribution in the optimal solution for the first part. If there are several budget allocations that satisfy the above condition then the algorithm allocates the maximal amount of budget that still satisfies the capacity and budget constraints.

Prior to formally describing our algorithm, we like to remind that $\chi_i$
corresponds to the characteristic vector of $a_i$, i.e., $(\chi_i)_i = 1$ and  $(\chi_i)_j = 0$ for every $j \neq i$. Accordingly, if $b \in \mathbb{N}^{n}$ is a budget allocation vector in which the $i$th coordinate represents the allocation to agent $a_i$, and $b_i=0$, then the allocation $b \vee k \chi_i$ corresponds to an update of $b$ by adding a budget $k$ to agent $a_i$. We say that the \emph{marginal value} of assigning $k$ units of budget to $a_i$ is $f(b \vee k \chi_i) - f(b)$, and the \emph{marginal value per unit} is $(f(b \vee k \chi_i)- f(b)) / k$.

\begin{algorithm}[t]
\DontPrintSemicolon
\SetKwInOut{Input}{Input}
\SetKwInOut{Output}{Output}
\Input{an online sequence of $n$ agents, a budget constraint $B \in \mathbb{N}_+$, capacity constraints $c \in \mathbb{N}_+^{n}$, a monotone submodular function $f: \mathbb{N}^{n} \rightarrow \mathbb{R}_+$, a parameter $\alpha \in \mathbb{R}_+$}
\Output{A budget allocation $b$} \smallskip

$b \leftarrow (0,0,\ldots,0)$

$f(b^L) \leftarrow$ value of the optimal budget allocation for agents in $L= \{a_i : t_i \leq 1/2\}$

\For {\emph{every agent} $a_i$ \emph{such that} $t_i \in (1/2, 1]$}
{
$K_i \leftarrow \big\{k \le c_i :~f({b} \vee k \chi_i) - f(b) \geq~\alpha k f({b}^L)/ B \big\}  \cup \big\{k 
+ \sum_{j \neq i} b_j \leq B \big\}$\\

\If{$K_i \neq \emptyset$}
{
$k \leftarrow \max_{k}\{K_i\}$\\
$b \leftarrow b \vee k \chi_i$
}
}
\textbf{return} $b$
\caption{Light Influence (LI)\label{alg: online}} 
\end{algorithm}

Having described our main component, we are now ready to complete the description of our algorithm. As already , we randomly combine algorithm LI with the classical algorithm for the secretary problem. Specifically, algorithm LI is employed with probability $5/8$ and the classical secretary algorithm with probability $3/8$. This latter algorithm assigns a maximal amount of budget to a single agent $a_i$ to attain an influence value of $f(c_i \chi_i)$. The algorithm selects $a_i$ by disregarding the first $n/e$ agents that arrive, and then picking the first agent whose influence value is better than any of the values of the first $n/e$ agents. This optimal algorithm is known to succeed in finding the single agent with the best influence with probability of $1/e$~\cite{Dynkin63}. 

\begin{algorithm}[t]
\DontPrintSemicolon
\SetKwInOut{Input}{Input}
\SetKwInOut{Output}{Output}
\Input{an online sequence of $n$ agents, a budget constraint $B \in \mathbb{N}_+$, capacity constraints $c \in \mathbb{N}_+^{n}$, a monotone submodular function $f: \mathbb{N}^{n} \rightarrow \mathbb{R}_+$, a parameter $\alpha \in \mathbb{R}_+$}
\Output{A budget allocation $b$} \smallskip

$r\leftarrow$ random number in $[0,1]$ \;
\uIf{$r \in [0,3/8]$}{$b \leftarrow$ run the classical secretary algorithm with $(n,B,c,f)$}
\ElseIf{$r \in (3/8,1]$}{$b \leftarrow$ run algorithm LI with $(n,B,c,f,\alpha)$}

\textbf{return} $b$
\caption{Online Influence Maximization \label{alg:onl_combine}}
\end{algorithm}

\subsection{Analysis}
We begin by analyzing the performance guarantee of algorithm LI, and later analyze the complete algorithm. Let $\opt^* =[\opt^*_1, \dots, \opt^*_n]$ be the optimal budget allocation for a given instance, and let $\opt^L$ be the budget allocation for the agents in $L$, that is, $\opt^L_i = \opt^*_i$ whenever $i \in L$ and $\opt^L_i = 0$, otherwise. Similarly, $\opt^R$ is the budget allocation for the agents in $R = [n] \setminus L$, i.e., $\opt^R_i = \opt^*_i$ for $i \in R$, and $\opt^R_i = 0$ for $i \notin R$. Recall that algorithm LI attends to the case in which no single agent has a significant influence contribution compared to the optimal value. More formally, let $\beta = \max_i f(c_i\chi_i) / f(\opt^*)$ be the ratio between the maximal contribution of a single agent and the optimal value.

\begin{lemma}\label{lem:lower_bound}
If $\alpha \geq 2\beta$ then $f(b) \geq \min\{\alpha f(\opt^L) / 2, f(\opt^R)- \alpha f(\opt^*)\}$.
\end{lemma}
\ifthenelse{\equal{\Context}{ABSTRACT}}
{} 
{
\begin{proof}
We prove this lemma by bounding the expected influence value of the algorithm in two cases and taking the minimum of them:

\smallskip \noindent {Case I: Algorithm LI allocates a budget of more than $B/2$ units.} We know that the algorithm attains a value of at least $\alpha f(b^L)/ B$ from each allocated budget unit by the selection rule $f(b \vee k_i\chi_i) - f(b) \geq \alpha k f(b^L)/ B$. Hence, the total influence of this allocation is at least
$$
f(b) > \frac{B}{2} \cdot \frac{\alpha f(b^L)}{B} = \frac{\alpha f(b^L)}{2} \geq \frac{\alpha f(\opt^L)}{2} \ .
$$

\smallskip \noindent {Case II: Algorithm LI allocates at most ${B/2}$ budget units.} We utilize the following lemma proven in~\cite[Lem~2.3]{SomaKIK14}.
\begin{lemma}
\label{lm: submod_property_2}
Let $f$ be a monotone submodular function over the integer lattice. For arbitrary $x,y$,
$$
f(x \vee y) \leq f(x) + \sum_{\substack{i\in [n]:\\ y_i > x_i}} \big(f(x \vee y_i\chi_i) - f(x)\big) \ .
$$~
\end{lemma}
This lemma applied to our case implies that 
\begin{equation} \label{online_2} 
f(b \vee \opt^R) \leq f(b) + \sum_{\substack{i\in [n]:\\\opt^R_i > b_i}}
\big(f(b \vee \opt^R_i \chi_i)-f(b)\big) \ .
\end{equation}

We consider two sub-cases:

\smallskip \noindent {Subcase A: There is $\ell \in [n]$ such that $\opt^R_{\ell}> B/2$.} Clearly, there can only be one agent $\ell$ having this property. One can easily validate that $f(b \vee \opt^R_\ell \chi_\ell)-f(b) \leq \beta \cdot f(\opt^*)$ by the definition of $\beta$ and Lemma~\ref{lm:submod_property_1}. Now, consider any agent $i \neq \ell$ with $\opt^R_i > b_i$. 
The reason that the optimal solution allocated more budget to $i$ than our algorithm cannot be the lack of budget since $\sum_i b_i < B/2$ and $\opt^R_{i}< B/2$. Hence, it must be the case that
\begin{equation}
\label{under_threshold}
\frac{f(b \vee \opt^R_i\chi_i)-f(b)}{\opt^R_i} < \alpha \frac{f(b^L)}{B} \ ,	
\end{equation}
by the selection rule of the algorithm. Note that $b$ in the above equation designates the budget allocation at the time that the agent $a_i$ was considered and not the final allocation. However, due to the weak version of marginal diminishing returns that was described in Lemma~\ref{lm:submod_property_1}, the inequality also holds for the final allocation vector. As a result,
\begin{align*}
	f(\opt^R) &\leq f(b \vee \opt^R)\\
	&\leq f(b) + \left(f(b \vee \opt^R_\ell \chi_\ell)-f(b)\right) + \sum_{\substack{i \in [n]\setminus \{\ell\}:\\ \opt^R_i > b_i}}
	\left(f(b \vee \opt^R_i \chi_i)-f(b)\right) \\ 
	&\leq f(b) + \beta f(\opt^*) + \alpha \frac{f(b^L)}{B} \cdot \frac{B}{2}\\ 
	&\leq f(b) + f(\opt^*) \cdot \left(\beta + \frac{\alpha}{2}\right)  , 
\end{align*}
where the first inequality follows due to the monotonicity of $f$, and the third inequality uses the sub-case assumption that there is one agent that receives at least half of the overall budget in $\opt^R$, and thus, $\sum_{i\neq \ell} \opt^R_{i} \leq B/2$. Recall that $\alpha \geq 2\beta$, and thus, $f(b) \geq f(\opt^R) - \alpha f(\opt^*)$.

\smallskip \noindent {Subcase B: $\opt^R_{i}\leq B / 2$, for every $i \in [n]$.} The analysis of this sub-case follows the same argumentation of the previous sub-case. Notice that for every agent $i \in [n]$ such that $\opt^R_i > b_i$, we can apply inequality~(\ref{under_threshold}). Consequently, we can utilize inequality~(\ref{online_2}), and get that 
$$
f(\opt^R) \leq f(b \vee \opt^R) \leq f(b) + \sum_{\substack{i \in [n]:\\ \opt^R_i > b_i}} \big(f(b \vee \opt^R_i \chi_i)-f(b)\big)
\leq f(b) + \alpha \frac{f(b^L)}{B} \cdot B  \ ,
$$
which implies that $f(b) \geq f(\opt^R) - \alpha f(\opt^*)$.~
\end{proof}
} 

Recall that we considered some arbitrary fixed ordering of the agents $\langle a_1,a_2,\dots,a_n \rangle$ that is not necessary their arrival order. Let $w_i$ the marginal contribution of agent $a_i$ to the optimal value when calculated according to this order. Namely, let $\opt^*_{< i} =[\opt^*_1, \dots, \opt^*_{i-1}, 0, \ldots, 0]$ be the allocation giving the same budget as $\opt^*$ for every agent $a_j$ with $j < i$, and $0$ for the rest, and define $w_i = f(\opt^*_{<i} \vee \opt^*_i\chi_i) - f(\opt_{<i})$. This point of view allow us to associate fixed parts of the optimal value to the agents in a way that is not affected by their order of arrival. Let $X_i$ be a random indicator for the event that $a_i \in L$, and let $W =\sum_{i=1}^n w_i X_i$. Let $\alpha \geq 2\beta$ to be determined later. 

By the weak version of marginal diminishing returns specified in  Lemma~\ref{lm:submod_property_1}, it holds that $f(\opt^L) \geq W$, and similarly, $f(\opt^R) \geq \sum_{i = 1}^n w_i(1 - X_i) = f(\opt^*) - W$. Using this observation, in conjunction with Lemma~\ref{lem:lower_bound}, we get that
$f(b) \geq \min\{\alpha W / 2, f(\opt^*) \cdot(1 - \alpha -{W}/{f(\opt^*)})\}$. Let $Y = W/ f(\opt^*)$, and observe that 
\begin{equation}
\label{lower_bound_2}
f(b) \geq f(\opt^*) \cdot \min\{ \alpha Y / 2, 1- \alpha -Y \} \ .
\end{equation}
Note that $Y \in [0,1]$ captures the ratio between the expected optimum value associated with the agents in $L$ and the (overall) optimum value. We continue to bound the expected value of $f(b)$ by proving the following lemma.
\begin{lemma}
\label{lem:expected}
Let $\alpha = 2/5$ and assume that $\beta \leq 1/5$, then, 
$$\mathbb{E}[f(b)] \geq \frac{f(\opt^*)}{20}\cdot \left(1-\sqrt{\beta}\right)^2.$$
\end{lemma}
\ifthenelse{\equal{\Context}{ABSTRACT}}
{} 
{
\begin{proof}
By assigning $\alpha = 2/5$ to the bound in inequality~\ref{lower_bound_2}, we obtain that
$$
f(b) \geq f(\opt^*) \cdot \min\left\{\frac{Y}{5}, \frac{3}{5}-Y\right\}.
$$
Notice that the expected value of $f(b)$ is
$$
\label{expected}
\mathbb{E}[f(b)] \geq f(\opt^*) \int_{0}^{\frac{3}{5}} [\mathrm{Pr}(Y \leq \gamma)]' \cdot \min \left\{ \frac{\gamma}{5}, \frac{3}{5}-\gamma \right\} d\gamma \ ,
$$
since $[\mathrm{Pr}(Y \leq \gamma)]'$ is the probability density function of $Y$. Now, observe that we can split the integral range into two parts
\begin{align}
\label{two_int}
\mathbb{E}[f(b)] &\geq {f(\opt^*)} \int_{0}^{\frac{1}{2}} [\mathrm{Pr}(Y \leq \gamma)]' \frac{\gamma}{5} d\gamma + f(\opt^*)\int_{\frac{1}{2}}^{\frac{3}{5}} [\mathrm{Pr}(Y \leq \gamma)]' \left(\frac{3}{5}-\gamma\right) d\gamma \nonumber \\ 
&\geq \frac{f(\opt^*)}{5}\int_{0}^{\frac{1}{2}} [\mathrm{Pr}(Y \leq \gamma)]' {\gamma} d\gamma.
\end{align}
To bound $\mathrm{Pr}(Y \leq \gamma)$, we use Chebyshev's inequality, while noting that 
$$
\mathbb{E}[Y] = \sum_{i=1}^n w_i \mathbb{E}[X_i] / f(\opt^*) = W / (2 f(\opt^*)) = 1/2 \ ,
$$ 
since $\mathbb{E}[X_i] = 1/2$ and $W = f(\opt^*)$. Now,
$$
\label{prob_Y}
\mathrm{Pr}\left[\Big|Y - \frac{1}{2}\Big| \geq c\right] \leq \frac{\mathrm{Var}[Y]}{c^2} \leq \frac{\beta}{4c^2} \ ,
$$
where the last inequality follows from \cite[Lem~B.5]{FeldmanNS11}. For completeness, the proof of this lemma appears as Lemma~\ref{cl:var_y} in the Appendix. Now, observe that $Y$ is symmetrically distributed around $1/2$, and therefore, $\mathrm{Pr}(Y \leq \frac 1 2 - c) = \mathrm{Pr}(Y \geq \frac 1 2 + c) \leq \beta / (8c^2)$. This implies that 
for every $\gamma \leq 1/2$,
$$
\mathrm{Pr}(Y \leq \gamma) \leq  \frac{\beta}{8(\frac{1}{2}-\gamma)^2}.
$$
Note that we cannot simply plug this upper bound on the cumulative distribution function into inequality~(\ref{two_int}). The fact that $Y$ is symmetrically distributed around $1/2$ implies that $\int_0^{1/2}\left[\mathrm{Pr}(Y \leq \gamma)\right]'d\gamma = 1/2$, and this does hold with this bound. To bypass this issue, and maintain the later constrain, we decrease the integration range. One can easily verify that  
$$
\int_0^{\frac{1 - \sqrt{\beta}}{2}} \left[\frac{\beta}{8(\frac{1}{2}-\gamma)^2}\right]'d\gamma =\frac 1 2 \ , 
$$
and as a result, we can infer that
\begin{align*}
\int_{0}^{\frac{1}{2}} [\mathrm{Pr}[Y \leq \gamma]]' {\gamma} d\gamma \geq 
\int_{0}^{\frac{1 - \sqrt{\beta}}{2}} \left[\frac{\beta}{8(\frac{1}{2}-\gamma)^2}\right]' {\gamma} d\gamma. 
\end{align*}
Specifically, this inequality holds since we essentially considered the worst distribution (from an algorithms analysis point of view) by shifting probability from higher values of $Y$ to smaller values (note that multiplication by $\gamma$). The proof of the lemma now follows since 
\begin{align*}
\mathbb{E}[f(b)] &\geq \frac{f(\opt^*)}{5} \int_{0}^{\frac{1 - \sqrt{\beta}}{2}} \left[\frac{\beta}{8(\frac{1}{2}-\gamma)^2}\right]'{\gamma} d\gamma 
\\ 
&= \frac{\beta \cdot f(\opt^*)}{20}\int_{0}^{\frac{1 - \sqrt{\beta}}{2}} \frac{\gamma}{(\frac{1}{2}-\gamma)^3}  d\gamma 
\\ 
&= \frac{\beta \cdot  f(\opt^*)}{20} \left[ \frac{4 \gamma - 1}{(1-2\gamma)^2}\right]^{\frac{1 - \sqrt{\beta}}{2}}_0 
\\ 
&=\frac{f(\opt^*)}{20}\cdot \left(1-\sqrt{\beta}\right)^2 \ .
\end{align*}
\end{proof}
} 
Recall that $\beta = \max_i f(c_i\chi_i)/ f(\opt^*)$. We next consider two cases depending on the value of $\beta$. When $\beta > 1/5$, our algorithm executes the classical secretary algorithm with probability $3/8$. This algorithm places a maximal amount of budget on the agent having maximum influence, $\max_i f(c_i\chi_i)$, with probability $1/e$. Consequently, 
$$
\mathbb{E}[f(b)] \geq \frac{3}{8}\cdot \frac{\beta \cdot f(\opt^*)}{e} > \frac{3f(\opt^*) }{40e} >  \frac{f(\opt^*) }{15e} \ .
$$
When $\beta \leq 1/5$, we know that our algorithm executes the classical secretary algorithm with probability $3/8$, and algorithm LI with probability $5/8$. Utilizing Lemma~\ref{lem:expected} results in
$$
\mathbb{E}[f(b)] \geq \frac{3}{8}\cdot \frac{\beta f(\opt^*)}{e} + \frac{5}{8}\cdot \frac{f(\opt^*)}{20}\cdot \left(1-\sqrt{\beta}\right)^2 
= \left(\frac{3\beta}{8e} + \frac{5}{160}\left(1-\sqrt{\beta}\right)^2 \right) \cdot f(\opt^*).
$$
One can validate that this latter term is minimized for $\beta = 1 /(12/e + 1)^2 \approx 0.034$, which implies that
$$
\mathbb{E}[f(b)] \geq \frac{3}{96+8e} f(\opt^*)> \frac{f(\opt^*)}{15e} \ .
$$
This completes the proof of the theorem.~
\end{proof} 

 

\bibliographystyle{plain}
\bibliography{BIP}

\begin{thebibliography}{10}

\bibitem{AlonGT12}
Noga Alon, Iftah Gamzu, and Moshe Tennenholtz.
\newblock Optimizing budget allocation among channels and influencers.
\newblock In {\em WWW}, pages 381--388, 2012.

\bibitem{BabaioffIKK07}
Moshe Babaioff, Nicole Immorlica, David Kempe, and Robert Kleinberg.
\newblock A knapsack secretary problem with applications.
\newblock In {\em APPROX}, pages 16--28, 2007.

\bibitem{BabaioffIKK08}
Moshe Babaioff, Nicole Immorlica, David Kempe, and Robert Kleinberg.
\newblock Online auctions and generalized secretary problems.
\newblock {\em SIGecom Exchanges}, 7(2), 2008.

\bibitem{BarmanUCM12}
Siddharth Barman, Seeun Umboh, Shuchi Chawla, and David~L. Malec.
\newblock Secretary problems with convex costs.
\newblock In {\em ICALP}, pages 75--87, 2012.

\bibitem{BateniHZ13}
MohammadHossein Bateni, Mohammad~Taghi Hajiaghayi, and Morteza Zadimoghaddam.
\newblock Submodular secretary problem and extensions.
\newblock {\em {ACM} Transactions on Algorithms}, 9(4):32, 2013.

\bibitem{BharathiKS07}
Shishir Bharathi, David Kempe, and Mahyar Salek.
\newblock Competitive influence maximization in social networks.
\newblock In {\em WINE}, pages 306--311, 2007.

\bibitem{BorgsBCL14}
Christian Borgs, Michael Brautbar, Jennifer~T. Chayes, and Brendan Lucier.
\newblock Maximizing social influence in nearly optimal time.
\newblock In {\em SODA}, pages 946--957, 2014.

\bibitem{Chen09}
Ning Chen.
\newblock On the approximability of influence in social networks.
\newblock {\em {SIAM} J. Discrete Math.}, 23(3), 2009.

\bibitem{ChenWW10}
Wei Chen, Chi Wang, and Yajun Wang.
\newblock Scalable influence maximization for prevalent viral marketing in
  large-scale social networks.
\newblock In {\em KDD}, pages 1029--1038, 2010.

\bibitem{ChenWY09}
Wei Chen, Yajun Wang, and Siyu Yang.
\newblock Efficient influence maximization in social networks.
\newblock In {\em KDD}, pages 199--208, 2009.

\bibitem{CohenDPW14}
Edith Cohen, Daniel Delling, Thomas Pajor, and Renato~F. Werneck.
\newblock Sketch-based influence maximization and computation: Scaling up with
  guarantees.
\newblock In {\em CIKM}, pages 629--638, 2014.

\bibitem{DemaineHMMRSZ14}
Erik~D. Demaine, MohammadTaghi Hajiaghayi, Hamid Mahini, David~L. Malec,
  S.~Raghavan, Anshul Sawant, and Morteza Zadimoghaddam.
\newblock How to influence people with partial incentives.
\newblock In {\em WWW}, pages 937--948, 2014.

\bibitem{DomingosR01}
Pedro Domingos and Matthew Richardson.
\newblock Mining the network value of customers.
\newblock In {\em KDD}, pages 57--66, 2001.

\bibitem{Dynkin63}
E.~B. Dynkin.
\newblock The optimum choice of the instant for stopping a markov process.
\newblock {\em Sov. Math. Dokl.}, 4:627–--629, 1963.

\bibitem{EftekharGK13}
Milad Eftekhar, Yashar Ganjali, and Nick Koudas.
\newblock Information cascade at group scale.
\newblock In {\em KDD}, pages 401--409, 2013.

\bibitem{Feige98}
Uriel Feige.
\newblock A threshold of ln {\it n} for approximating set cover.
\newblock In {\em J. ACM}, pages 634--652, 1998.

\bibitem{FeldmanNS11}
Moran Feldman, Joseph Naor, and Roy Schwartz.
\newblock Improved competitive ratios for submodular secretary problems.
\newblock In {\em APPROX}, pages 218--229, 2011.

\bibitem{GoldenbergLM2001}
Jacob Goldenberg, Barak Libai, and Eitan Muller.
\newblock Talk of the network: A complex systems look at the underlying process
  of word-of-mouth.
\newblock {\em Marketing Letters}, 2001.

\bibitem{Gomez-RodriguezLK12}
Manuel Gomez{-}Rodriguez, Jure Leskovec, and Andreas Krause.
\newblock Inferring networks of diffusion and influence.
\newblock {\em {TKDD}}, 5(4):21, 2012.

\bibitem{GoyalBL10}
Amit Goyal, Francesco Bonchi, and Laks V.~S. Lakshmanan.
\newblock Learning influence probabilities in social networks.
\newblock In {\em WSDM}, pages 241--250, 2010.

\bibitem{GoyalK12}
Sanjeev Goyal and Michael Kearns.
\newblock Competitive contagion in networks.
\newblock In {\em STOC}, pages 759--774, 2012.

\bibitem{Granovetter1978}
M.~Granovetter.
\newblock Threshold models of collective behavior.
\newblock {\em Am. J. Sociol.}, 83(6):1420--1443, 1978.

\bibitem{GuptaRST10}
Anupam Gupta, Aaron Roth, Grant Schoenebeck, and Kunal Talwar.
\newblock Constrained non-monotone submodular maximization: Offline and
  secretary algorithms.
\newblock In {\em WINE}, pages 246--257, 2010.

\bibitem{HartlineMS08}
Jason~D. Hartline, Vahab~S. Mirrokni, and Mukund Sundararajan.
\newblock Optimal marketing strategies over social networks.
\newblock In {\em WWW}, pages 189--198, 2008.

\bibitem{HeK13}
Xinran He and David Kempe.
\newblock Price of anarchy for the n-player competitive cascade game with
  submodular activation functions.
\newblock In {\em WINE}, pages 232--248, 2013.

\bibitem{KempeKT03}
David Kempe, Jon~M. Kleinberg, and {\'{E}}va Tardos.
\newblock Maximizing the spread of influence through a social network.
\newblock In {\em KDD}, pages 137--146, 2003.

\bibitem{KempeKT05}
David Kempe, Jon~M. Kleinberg, and {\'{E}}va Tardos.
\newblock Influential nodes in a diffusion model for social networks.
\newblock In {\em ICALP}, pages 1127--1138, 2005.

\bibitem{KhannaL14}
Sanjeev Khanna and Brendan Lucier.
\newblock Influence maximization in undirected networks.
\newblock In {\em SODA}, pages 1482--1496, 2014.

\bibitem{Kleinberg07}
Jon Kleinberg.
\newblock Cascading behavior in networks: Algorithmic and economic issues.
\newblock {\em Algorithmic Game Theory}, 2007.

\bibitem{Kleinberg05}
Robert~D. Kleinberg.
\newblock A multiple-choice secretary algorithm with applications to online
  auctions.
\newblock In {\em SODA}, pages 630--631, 2005.

\bibitem{LeiMMCS15}
Siyu Lei, Silviu Maniu, Luyi Mo, Reynold Cheng, and Pierre Senellart.
\newblock Online influence maximization.
\newblock In {\em SIGKDD}, pages 645--654, 2015.

\bibitem{MaeharaYK15}
Takanori Maehara, Akihiro Yabe, and Ken{-}ichi Kawarabayashi.
\newblock Budget allocation problem with multiple advertisers: A game theoretic
  view.
\newblock In {\em ICML}, pages 428--437, 2015.

\bibitem{MathioudakisBCGU11}
Michael Mathioudakis, Francesco Bonchi, Carlos Castillo, Aristides Gionis, and
  Antti Ukkonen.
\newblock Sparsification of influence networks.
\newblock In {\em KDD}, pages 529--537, 2011.

\bibitem{MosselR07}
Elchanan Mossel and Sebastien Roch.
\newblock On the submodularity of influence in social networks.
\newblock In {\em STOC}, pages 128--134, 2007.

\bibitem{Nash50}
John~F. Nash.
\newblock Equilibrium points in $n$-person games.
\newblock {\em Proc. Natl. Acad. Sci.}, 36:48--49, 1950.

\bibitem{RichardsonD02}
Matthew Richardson and Pedro Domingos.
\newblock Mining knowledge-sharing sites for viral marketing.
\newblock In {\em KDD}, pages 61--70, 2002.

\bibitem{SeemanS13}
Lior Seeman and Yaron Singer.
\newblock Adaptive seeding in social networks.
\newblock In {\em FOCS}, pages 459--468, 2013.

\bibitem{ShihL2001}
Hsu-Shih Shih and E.~Stanley Lee.
\newblock {\em Discrete Multi-Level Programming in a Dynamic Environment},
  pages 79--98.
\newblock Physica-Verlag HD, 2001.

\bibitem{Singer12}
Yaron Singer.
\newblock How to win friends and influence people, truthfully: influence
  maximization mechanisms for social networks.
\newblock In {\em WSDM}, pages 733--742, 2012.

\bibitem{SomaKIK14}
Tasuku Soma, Naonori Kakimura, Kazuhiro Inaba, and Ken{-}ichi Kawarabayashi.
\newblock Optimal budget allocation: Theoretical guarantee and efficient
  algorithm.
\newblock In {\em ICML}, pages 351--359, 2014.

\bibitem{Sviridenko04}
Maxim Sviridenko.
\newblock A note on maximizing a submodular set function subject to a knapsack
  constraint.
\newblock {\em Oper. Res. Lett.}, 32(1):41--43, 2004.

\bibitem{Vetta02}
Adrian Vetta.
\newblock Nash equilibria in competitive societies, with applications to
  facility location, traffic routing and auctions.
\newblock In {\em FOCS}, pages 416--425, 2002.

\bibitem{YangMPH16}
Yu~Yang, Xiangbo Mao, Jian Pei, and Xiaofei He.
\newblock Continuous influence maximization: What discounts should we offer to
  social network users?
\newblock In {\em SIGMOD}, pages 727--741, 2016.

\end{thebibliography}

\ifthenelse{\equal{\Context}{ABSTRACT}}
{} 
{
\appendix
\section{Additional details}
The following claim was established by~\cite[Lem~B.5]{FeldmanNS11}. We include it here for completeness. 
\begin{lemma} \label{cl:var_y}
Consider the random variable $Y = \sum_i w_i X_i/f(\opt^*)$, defined in the proof of Theorem~\ref{thm:online_main}. Its variance is $\mathrm{Var}[Y] \leq \beta/ 4$.
\end{lemma}
\begin{proof}
\begin{align*}
Var[Y]&=  \frac{Var[\sum_i w_i X_i]}{f^2(\opt^*)} = \frac{\sum_i w_i^2 Var[X_i]}{f^2(\opt^*)} = \frac{\sum_i w_i^2}{4f^2(\opt^*)} \\
&\leq  \frac{\max_i w_i \cdot \sum_i w_i}{4f^2(\opt^*)} =\frac{\max_i w_i \cdot f(\opt^*)}{4f^2(\opt^*)} \leq \frac{\beta}{4} \ .
\end{align*}~
\end{proof}
} 

\end{document}